\documentclass[conference]{IEEEtran}
\IEEEoverridecommandlockouts
\usepackage{cite}
\usepackage{amsmath,amssymb,amsfonts}
\usepackage{algorithmic}
\usepackage{graphicx}
\usepackage{textcomp}
\usepackage{xcolor}

\usepackage{float}
\usepackage{mathtools}
\usepackage{algorithm}
\usepackage{amsthm}
\usepackage{tablefootnote}
\usepackage{threeparttable}

\usepackage{multirow}
\renewcommand{\baselinestretch}{0.987} 
\usepackage{footmisc}

\floatname{algorithm}{Algorithm}

\newcommand{\be}{\begin{equation}}
\newcommand{\ee}{\end{equation}}
\newcommand{\ba}{\begin{array}}
	\newcommand{\ea}{\end{array}}
\newcommand{\nid}{\noindent}

\newcommand{\m}{\hspace{-.05cm}}
\renewcommand{\thesubsubsection}{\thesubsection.\arabic{subsubsection}}
\newtheorem{theorem}{Theorem}

\newcommand\floor[1]{\lfloor#1\rfloor}
\newcommand\ceil[1]{\lceil#1\rceil}
\newcommand{\argmin}{\operatornamewithlimits{argmin}}
\newcommand\inlineeqno{\stepcounter{equation}\ (\theequation)}

\renewcommand*{\thefootnote}{\fnsymbol{footnote}}

\def\BibTeX{{\rm B\kern-.05em{\sc i\kern-.025em b}\kern-.08em
		T\kern-.1667em\lower.7ex\hbox{E}\kern-.125emX}}
\begin{document}
	
	\title {Fast Decoder for Overloaded Uniquely Decodable \\  Synchronous Optical CDMA}
	
	\author{ \IEEEauthorblockN{Michel Kulhandjian$^{\dag}$, Hovannes Kulhandjian$^{\ddag}$, Claude D'Amours$^{\dag}$, Halim Yanikomeroglu$^{\dag\dag}$, Gurgen Khachatrian$^{\dag\dag\dag}$
		}\\
		\IEEEauthorblockA{$^{\dag}$School of Electrical Engineering and Computer Science, University of Ottawa, 
			Ottawa, Ontario, K1N 6N5, Canada\\
			E-mail: \texttt{mkk6@buffalo.edu,cdamours@uottawa.ca} \\
			$^{\ddag}$Department of Electrical and Computer Engineering,
			California State University, Fresno,
			Fresno, CA 93740, U.S.A. \\
			E-mail: \texttt{hkulhandjian@csufresno.edu} \\
			$^{\dag\dag}$Systems and Computer Engineering,
			Carleton University,
			Ottawa, ON, K1S 5B6, Canada \\
			E-mail: \texttt{Halim.Yanikomeroglu@sce.carleton.ca}\\
			$^{\dag\dag\dag}$College of Science \& Engineering,
			American University of Armenia, Yerevan 0019, Republic of Armenia \\
			E-mail: \texttt{gurgenkh@aua.am}} 
		\vspace{-0.9 cm}  }
	
	\maketitle
	
	\begin{abstract}
		In this paper, we propose a fast decoder algorithm for uniquely decodable (errorless) code sets for overloaded synchronous optical code-division multiple-access (O-CDMA) systems. The proposed decoder is designed in a such a way that the users can uniquely recover the information bits with a very simple decoder, which uses only a few comparisons. Compared to maximum-likelihood (ML) decoder, which has a high computational complexity for even moderate code lengths, the proposed decoder has much lower computational complexity. Simulation results in terms of bit error rate (BER) demonstrate that the performance of the proposed decoder for a given BER requires only $1-2$ dB higher signal-to-noise ratio (SNR) than the ML decoder.
	\end{abstract}
	
	\begin{IEEEkeywords}
		Optical code-division multiple-access, uniquely decodable (errorless) codes, overloaded synchronous O-CDMA.
	\end{IEEEkeywords}
	
	\section{{Introduction}}
	\renewcommand*{\thefootnote}{\arabic{footnote}}
	Optical code-division multiple-access (O-CDMA) has recently received substantial attention in local area networks (LAN's) where the traffic tends to be bursty \cite{heritage2007}. This is due to the development of large bandwidth fiber-optic communication channels, which introduce several advantages over the conventional networking. One of the key features of O-CDMA is that it allows simultaneous users transmit data asynchronously, with no waiting time through the assignment of the unique sequence code. 
	
	O-CDMA can also be applied in free-space optics (FSO) \cite{gibson2004},   \cite{alzenad2018}. Unlike the fiber-optic communications in which data is transmitted by propagation of light through a fiber in FSO the optical beams are sent through free air \cite{malik2015}.
	
	Conventional CDMA signature codes $\{\pm 1\}$ designed for the radio systems are no longer suitable when they are used in optical systems because of the unipolarity of the incoherently detected signals. Although coherent signal processing in O-CDMA is possible in principle, it is not practical, as it is fairly difficult to maintain  the correct phase of the optical carrier at high frequency. Hence, several unipolar $\{0,1\}$ signature codes have been proposed for incoherent O-CDMA systems, such as algebraic construction of a new family of optical orthogonal codes (OOC) \cite{Fan1989} for use in asynchronous CDMA fiber-optic communication systems. The optical orthogonal code $\mathbf{C}$ is a family of unipolar codes characterized by a quadruple $(n, \omega, \lambda_a, \lambda_c)$, where $n$ denotes the code length, $\omega$ denotes its weight (the number of ones), and $\lambda_a$ and $\lambda_c$ denote the maximum value of the out-of-phase auto-correlation and maximum value of the cross-correlation, respectively. Those two properties are defined as follows:
	
	\textit{1) The Autocorrelation Property:}\vspace{-0.1cm}
	\begin{equation}
	\sum_{t=0}^{n-1} x_t x_{t\oplus\tau} \leq \lambda_a,
	\end{equation}
	for any $\mathbf{x} \in \mathbf{C}$ and any integer $\tau$, $0 < \tau < n$, where $\oplus $ denotes modulo $n$ addition.
	
	\textit{2) The Cross-correlation Property:}\vspace{-0.1cm}
	\begin{equation}
	\sum_{t=0}^{n-1} x_t y_{t\oplus \tau} \leq \lambda_c,
	\end{equation}
	for any $\mathbf{x} \neq \mathbf{y} \in \mathbf{C}$ and any integer $\tau$.
	
	Modified prime codes \cite{Maric1993}, \cite{Svetislav1993} have been proposed for the synchronous O-CDMA systems. Ideal orthogonality between the sequences cannot be obtained even in the synchronous case, as the signals are unipolar in nature, that is, two signals cannot be added up to get zero. This incoherent processing renders codes with good correlation properties, but unequal number of ``ones'' such as Gold sequences, which are not applicable in fiber-optic CDMA systems. Therefore, multiple-access interference (MAI) limits the performance of incoherent O-CDMA systems.
	
	A number of detectors have been proposed to alleviate MAI in the O-CDMA system. Among the detectors are single user detectors such as correlator detector, hard limiter correlator, chip level detector, etc. Verdu \cite{Verdu1999} has proposed an optimal detector, which can be used in O-CDMA systems. However, the optimal detector is exponentially proportional to the number of users and that prohibits its practical implementation. Various suboptimal detection techniques have been proposed with low complexity. These suboptimal approaches can be classified into two categories: linear and non-linear multiuser detectors. Among the linear multiuser detectors are matched filter (MF), minimum-mean-square-error (MMSE), etc. In a non-linear subtractive interference cancellation detector the interference is first estimated and then it is subtracted from the received signal before performing the detection. The cancellation process can be carried out either through successively interference cancellation (SIC) \cite{Kobayashi2001}, or through parallel interference cancellation (PIC) \cite{Guo2000}, \cite{Xue1999}. In non-linear iterative detectors \cite{Wang1999} - \cite{Sasipriya2014} and in probabilistic data association (PDA) \cite{Romano2005} the aim is to suppress the MAI in each iteration in order to improve the overall error performance. Suboptimal polynomial time detectors that are based on the geometric approach are studied in \cite{Manglani2006} - \cite{Najkha2005}. 
	
	Due to the large demand for spectral efficiency in the Internet traffic, in this paper, we consider the overloaded O-CDMA case, where the number of users, K, is greater than the spreading factor, L. In general, the users’ signals cannot be separated by either linear or nonlinear detectors in overloaded systems even in the case of asymptotically vanishing noise. We therefore seek to design spreading codes such that decoding can achieve asymptotically zero probability of error multiuser detection when the signal-to-noise ratio (SNR) becomes arbitrary large. The uniquely decodable (UD) class of codes that guarantee ``errorless'' communication in an ideal (noiseless) synchronous O-CDMA also show a good performance in the presence of noise. Finding those overloaded UD class of codes for noiseless channel is directly related to coin-weighing problem, one of Erd\"{o}s' problems in \cite{Erdos1963}. It is a special case of a general problem and in literature \cite{Shapiro1963} - \cite{Mow2009} authors used the term detecting matrices. Lindstr\"{o}m in \cite{Lindstrom1965} defines the same problem as the detecting set of vectors. Given an integer $q \geq 2$ and
	a finite set alphabet $M$ of rational integers, let $\mathbf{v}_i$ for $1\leq i \leq n$ be $L$-dimensional (column) vectors with all components from $M$ such that the $q^n$ sums\vspace{-0.0cm}
	\begin{equation}\vspace{-0.0cm}
	\sum_{i = 1}^n  \epsilon_i \mathbf{v}_i  \: \: (\epsilon_i = 0,1,2,\dots, q-1)
	\end{equation}
	\nid are all distinctly unique, then $\{ \mathbf{v}_1, \dots, \mathbf{v}_n \}$ are detecting set of vectors. Let $F_q(L)$ be the maximal number of $L$-dimensional vectors forming a detecting set. Let $f_q(n)$ be the minimal vector length for a given $n$ number of vectors. The problem of determining the  $f_q(n)$ as a special case when $q=2$, $M = \{0,1\}$ that can be equivalently expressed as a coin-weighing problem: what is the minimal number of weighings on an accurate scale to determine all false coins in a set of $n$ coins. The choice of coins
	for a weighing must not depend on results of previous weighings. This problem was first introduced by H. S. Shapiro \cite{Shapiro1963} for $n=5$. Only few cases of $n$ was proved the minimal number of $L$, however
	for larger $n$ the $f_2(n)$ has been estimated \cite{Erdos1970}. Lindstr\"{o}m in \cite{Lindstrom1965} gives an explicit construction of $L \times \gamma(L+1)$ binary (alphabet $\{0, 1\}$) and $L \times \gamma(L)+1$ antipodal (alphabet $\{\pm 1\}$) detecting matrices, where $\gamma(L)$ is the number of ones in the binary expansion of all positive integers less than $L$. 
	He also proved that the lower bound in the case of $M=\{0,1\}$ or $\{\pm1\}$ is\vspace{-0.0cm}
	\begin{equation}\vspace{-0.0cm}
	\lim_{n\to\infty} \frac{f_2(n) \log{n}}{n} =2.
	\end{equation}
	Cantor and Mills \cite{Cantor1966} constructed a class of $2^k \times (k+2)2^{(k-1)}$ ternary (alphabet $\{0, \pm 1\}$) detecting matrices for $k \in \mathbb{Z}^+$, which implies that in the case of $M=\{0,\pm1\}$ the lower bound is \vspace{-0.0cm}
	\begin{equation}\vspace{-0.0cm}
	\lim_{n\to\infty} \frac{f_3(n) \log{n}}{n} \leq 2.
	\end{equation}\vspace{0.0cm}
	In the literature, most of the explicit construction algorithms of UD code sets are recursive. It is worth mentioning that the maximum number 
	of vectors of the explicit constructions of binary, antipodal and ternary code sets are $K_{max}^b = \gamma(L+1)$, $K_{max}^a = \gamma(L) +1$ and $K_{max}^t = (k+2)2^{(k-1)}$,
	as shown in Table \ref{table:binary}, Table \ref{table:antipodal} and Table \ref{table:ternary}, respectively. The applications of such codes varies but typically is mostly seen in the noiseless transmission channels. As an example, they can be suitable for the multi-access adder channels \cite{Chang1979} - \cite{Chang1984} and in wired communications, which can double (or more) the bandwidth at modest/moderate expense of the increase in computational cost. The authors in \cite{Marvasti2009} and \cite{Alishahi2012} motivate the overloaded binary UD code sets for the application in O-CDMA fiber-optic communications. 
	
	To the best of our knowledge, there are no known explicit constructions that generates larger than $K_{max}^b$, $K_{max}^a$ and $K_{max}^t$ vectors for a given $L$ in a code set. Several authors have proposed linear decoders in the noiseless scenarios for their explicit construction achieving $K_{max}$. For example, Martirossian and Khachatrian in
	\cite{Khachatrian1989} presented a linear decoder for their explicit construction of binary code sets in the noiseless transmission channel. The linear decoders corresponding 
	to their explicit construction with $K_{max}$ of antipodal and ternary code sets can be found in \cite{Khachatrian1995} - \cite{michel2018} and \cite{Chang1979} - \cite{Khachatrian1998}, respectively. Although such overloaded UD code sets theoretically facilitate a large capacity, their decoding for noisy transmission has always been a greater challenge to deal with. For noisy channel, the proposed decoders stand ineffective to provide an acceptable
	error performance. In general, the efficiency of the whole system is determined by the decoder, which must have a simple design and perform comparably better in a noisy transmission channels.	In fact, in noisy channels, those code sets that have $K_{max}$ vectors need a maximum-likelihood (ML) decoder to determine the received vector,
	a process which is considered NP-hard \cite{Lupas1989}. 
	
	Recently, in \cite{Marvasti2013}, an overloaded synchronous O-CDMA based on unipolar Walsh code (UWC) for fiber-optic communication system was proposed. The number of users of their construction is $K = 2^{k+1}-2$ for the given length $L= 2^k$. This exceeds the maximum $K_{max}^b$ when $k<4$. However, for the cases when $k\geq 4$ the number of users becomes $K < K_{max}^b$, since in their proposed system the users are divided into two groups. Users of each group transmit at the same power level different from the level of the other group's users. Unlike all the UD codes presented in Tables \ref{table:binary}, \ref{table:antipodal} and \ref{table:ternary}, where each users transmit at the same power level. Based on UWC properties the authors present a simple receiver that alleviates MAI completely.
	
	\vspace{-0.0cm}
	\begin{table*}[h]
		\caption{Binary Codes} \vspace{-0.1cm}
		\centering 
		\begin{threeparttable}
			\begin{tabular}{l l c c c c} 
				\hline\hline  
				\multicolumn{1}{c}{\multirow{2}{*}[-1.5pt]{\bf{Year}}} &
				\multicolumn{1}{l}{\multirow{2}{*}[-1.5pt]{\bf{Authors and Publications}}}  & \multicolumn{1}{c}{\multirow{2}{*}[-1.5pt]{$\bf{n}$}} & \multicolumn{1}{c}{\multirow{2}{*}[-1.5pt]{$\bf{K}$}} & \multicolumn{2}{c}{\multirow{1}{*}[-1.5pt]{\bf{Decoder}}} \\[0.5ex]  \cline{5-6}
				&  &  &  & \multirow{1}{*}[-1.5pt]{\bfseries{Noiseless}} & \multirow{1}{*}[-1.5pt]{\bfseries{AWGN}}\\ [1.0ex]
				\hline   \rule{-3pt}{2.5ex} 
				1963 & S\"{o}derberg and Shapiro \cite{Shapiro1963} & $L$  & $<\gamma(L+1)$ & No  & No\\[0.6ex]
				1964 & Lindstr\"{o}m  \cite{Lindstrom1964}  & $L$ & $\bf{\boldsymbol{\gamma}(L+1)}$\tnote{\dag} & No & No \\[0.6ex]
				1966 & Cantor and Mills \cite{Cantor1966}  & $2^k-1$ & $\bf{k2^{(k-1)}}$ & No & No \\[0.6ex]
				1989 & Martirossian and Khachatrian \cite{Khachatrian1989}  & $L$ & $\bf{\boldsymbol{\gamma}(L+1)}$ & Yes & No \\[0.6ex]
				\hline 
			\end{tabular}
			\footnotesize
			\begin{tablenotes}
				\item[\dag] Code set constructions that achieve the maximum number of vectors $\mathbf{K}_{max}$ are presented in bold.
			\end{tablenotes}
		\end{threeparttable}
		
		\label{table:binary}
	\end{table*}
	\vspace{-0.0cm}
	\begin{table*}[h]
		\caption{Antipodal Codes} \vspace{-0.1cm}
		\centering 
		\begin{tabular}{l l c c c c} 
			\hline\hline  
			\multicolumn{1}{c}{\multirow{2}{*}[-1.5pt]{\bf{Year}}} &
			\multicolumn{1}{l}{\multirow{2}{*}[-1.5pt]{\bf{Authors and Publications}}}  & \multicolumn{1}{c}{\multirow{2}{*}[-1.5pt]{$\bf{n}$}} & \multicolumn{1}{c}{\multirow{2}{*}[-1.5pt]{$\bf{K}$}} & \multicolumn{2}{c}{\multirow{1}{*}[-1.5pt]{\bf{Decoder}}} \\[0.5ex]  \cline{5-6}
			&  &  &  & \multirow{1}{*}[-1.5pt]{\bfseries{Noiseless}} & \multirow{1}{*}[-1.5pt]{\bfseries{AWGN}}\\ [1.0ex]
			\hline   \rule{-3pt}{2.5ex} 
			1964 & Lindstr\"{o}m \cite{Lindstrom1964} & $L$  & $\bf{\boldsymbol{\gamma}(L)+1}$ & No  & No\\[0.6ex]
			1987 & Khachatrian and Martirossian \cite{Khachatrian1987}  & $L$ & $\bf{\boldsymbol{\gamma}(L)+1}$ & No & No \\[0.6ex]
			1995 & Khachatrian and Martirossian \cite{Khachatrian1995}  & $2^k$ & $\bf{k2^{(k-1)}+1}$ & Yes & No \\[0.6ex]
			2012 & Kulhandjian and Pados \cite{michel2012}  & $2^k$ & $\bf{k2^{(k-1)}+1}$ & Yes & No \\[0.6ex]
			2018 & Kulhandjian \textit{et al.} \cite{michel2018}  & $2^k$ & $\bf{k2^{(k-1)}+1}$ & Yes & Yes \\[0.6ex]
			\hline 
		\end{tabular}
		\label{table:antipodal}
	\end{table*}
	\vspace{-0.0cm}
	\begin{table*}[h]
		\caption{Ternary Codes} \vspace{-0.1cm}
		\centering 
		\begin{tabular}{l l c c c c} 
			\hline\hline  
			\multicolumn{1}{c}{\multirow{2}{*}[-1.5pt]{\bf{Year}}} &
			\multicolumn{1}{l}{\multirow{2}{*}[-1.5pt]{\bf{Authors and Publications}}}  & \multicolumn{1}{c}{\multirow{2}{*}[-1.5pt]{$\bf{n}$}} & \multicolumn{1}{c}{\multirow{2}{*}[-1.5pt]{$\bf{K}$}} & \multicolumn{2}{c}{\multirow{1}{*}[-1.5pt]{\bf{Decoder}}} \\[0.5ex]  \cline{5-6}
			&  &  &  & \multirow{1}{*}[-1.5pt]{\bfseries{Noiseless}} & \multirow{1}{*}[-1.5pt]{\bfseries{AWGN}}\\ [1.0ex]
			\hline   \rule{-3pt}{2.5ex} 
			1966 & Cantor and Mills \cite{Cantor1966} & $2^k$  & $\bf{(k+2)2^{(k-1)}}$ & No  & No\\[0.6ex]
			1979 & Chang and Weldon \cite{Chang1979}  & $2^k$ & $\bf{(k+2)2^{(k-1)}}$ & Yes & No \\[0.6ex]
			1982 & Ferguson \cite{Ferguson1982}   & $2^k$ & $\bf{(k+2)2^{(k-1)}}$ & Yes & No \\[0.6ex]
			1984 & Chang \cite{Chang1984}  & $2^k$ & $\bf{(k+2)2^{(k-1)}}$ & No & No \\[0.6ex]
			1998 & Khachatrian and Martirossian \cite{Khachatrian1998}  & $2^k$ & $\bf{(k+2)2^{(k-1)}}$ & Yes & No \\[0.6ex]
			2012 & Mashayekhi and Marvasti \cite{Marvasti2012}  & $2^k$ & $2^{(k+1)}-1$ & Yes & Yes \\[0.6ex]
			2016 & Singh \textit{et al.} \cite{Marvasti2016}   & $2^k$ & $2^{(k+1)}-2$ & Yes & Yes \\[0.6ex]
			2018 & Kulhandjian \textit{et al.} \cite{michel22017}   & $2^k$ & $2^{(k+1)}+2^{(k-2)}-1$ & Yes & Yes \\[0.6ex]
			\hline 
		\end{tabular}
		\label{table:ternary}
	\end{table*}
	
	In this work, for the first time we consider the problem of designing a fast decoder for binary UD code sets, which achieve maximum number of users $K_{max}^b$ presented in \cite{Khachatrian1989}. These recursive construction sets provide one possible constructs out of all UD code set, which are distinct from other known constructs, shown in Table \ref{table:binary}. The proposed decoder is designed in such a way that the user can uniquely recover the information bits with a very simple decoder, which uses only a few comparisons. In contrast to ML decoder, the proposed decoder has much lower computational complexity. Simulation results in terms of bit error rate (BER) demonstrate that the performance of the proposed decoder for a given BER requires only $1-2$ dB higher SNR than the ML decoder.
	
	The rest of the paper is organized as follows. In Section \ref{sec:transmission}, we discuss fiber-optic transmission and assumptions made, followed by the errorless code set construction in Section \ref{sec:construction}. The minimum distance of such code sets is presented in Section \ref{minDist} followed by the noiseless decoding algorithm (NDA) in Section \ref{NDAalgorithm} and fast decoding algorithm (FDA) in Section \ref{fastDecoder}. The complexity analysis is presented in Section \ref{performanceAnalysis}. After illustrating simulation results in Section \ref{simulation}, a few conclusions are drawn in Section \ref{conclusion}.
	
	The following notations are used in this paper. All boldface lower
	case letters indicate column vectors and upper case letters indicate
	matrices, $()^T$ denotes transpose operation, $\mathsf{sgn}$ denotes the sign function, $| . |$ is the scalar magnitude, $|| . ||$ is vector 
    norm, $\oplus$ is the modulo $2$ addition, $\otimes$ is the Kronecker product, $\ceil{.}$ is the ceiling function and $\floor{.}$ is the flooring function.
	\vspace{-0.2cm}
	\section{Fiber-Optic Transmission and Assumptions}
	\label{sec:transmission}
	In the present work, chip synchronization among all the transmitters is assumed. This provides a worst-case estimate of the performance of what is in reality a fully asynchronous system, which only requires chip synchronization between the source transmitter and the target receiver.
	In O-CDMA system using unipolar codes transmission takes place in two main forms. In multiple-access system each user performs on-off keying (OOK), where the bit ``$1$'' is represented by the presence of the codeword while bit ``$0$'' is represented by its absence. In other words, $i$-th user's bit ``$1$'' is the spreading code $\mathbf{c}_i$ and bit ``$0$'' is $\mathbf{0}$. Mathematically, we can formulate the system model as\vspace{-0.2cm}
		\begin{eqnarray}
	\label{system_OOK} \mathbf{y}_1 &=& \sum_{i = 1}^K \mathbf{c}_i x_i  \\  &=& \mathbf{C}\mathbf{x},
	\end{eqnarray}
	\nid where $x_i \in \{0,1\}$ is the $i$-th user's information bit and $\mathbf{C} \in \{0,1\}^{L \times K}$ is the spreading code matrix. The other technique is the bipolar signaling, where the $i$-th user's  bits ``$1$'' and ``$0$'' are represented by the spreading code $\mathbf{c}_i \in \{0,1\}^{K\times 1}$ and $\mathbf{\bar{c}}_i$ binary complement\footnote{Binary complement is defined as $\mathbf{\bar{c}} = \mathbf{c} \oplus \mathbf{J}_{K,1}$.} of $\mathbf{c}_i$, which can be written as\vspace{-0.0cm}
	\begin{eqnarray}
	\label{system_BPS} \mathbf{y}_2 &=& \sum_{i = 1}^K \frac{1}{2} ( \mathbf{c}_i' x_i' + \mathbf{J}_{ L, 1}) \\  &=& \frac{1}{2} (\mathbf{C}'\mathbf{x}'+ K\mathbf{J}_{L, 1}),
	\end{eqnarray}
	\nid where $x_i' = 2 x_i -1 \in \{\pm 1\}$,  $\mathbf{C}' = 2 \mathbf{C} -\mathbf{J}_{L, K} \in \{\pm 1\}^{L \times K}$ and $\mathbf{J}_{t,m}$ matrix contains $(t \times m)$ ones. Notice that bipolar signaling in (\ref{system_BPS}) is an affine transformation of the antipodal system discussed in \cite{michel2018}. In case the bipolar signaling is involved in the transmission system we can rewrite (\ref{system_BPS}) as follows: \vspace{-0.0cm}
	\vspace{-0.0cm}
	\begin{equation}
	\label{system_MOOK} \mathbf{r}_1 = 2\mathbf{y}_2 -K\mathbf{J}_{L, 1} = \mathbf{C}' \mathbf{x}',
	\end{equation}
	\nid where $\mathbf{x}' = 2\mathbf{x}-\mathbf{J}_{K,1} \in \{\pm 1\}^{K\times 1}$. It can be shown that solving (\ref{system_MOOK}) is equivalent to solving (\ref{system_BPS}). Noiseless decoding algorithm in \cite{michel2012} potentially can be used to solve (\ref{system_MOOK}) and in the case of AWGN channel the fast decoding algorithm recently presented in \cite{michel2018} can be utilized. 
	Therefore, in this manuscript we will employ the OOK transmission scheme instead of the bipolar.
	
	\vspace{0.1cm}
	\section{Errorless Code Set Construction}
	\label{sec:construction}
	We recall that a binary code set $\mathbf{C} \in \{0, 1\}^{L \times K}$ is uniquely decodable over signals $\mathbf{x} \in \{ \pm 1 \}^{K\times 1}$ or $\mathbf{x} \in \{ 0, 1 \}^{K\times 1}$,  \textit{if and only if} for any $\mathbf{x}_1 \neq \mathbf{x}_2$, $\mathbf{C}\mathbf{x}_1 \neq \mathbf{C}\mathbf{x}_2$ or, equivalently, $\mathbf{C}(\mathbf{x}_1 - \mathbf{x}_2)\neq \mathbf{0}_{L \times 1}$ \cite{michel2012}. We can rewrite the unique decodability necessary and sufficient condition as 
	\begin{equation}
	\label{null01}
	\mathsf{Null}(\mathbf{C}) \cap \{0,\pm 1\}^{K\times 1} = \{ 0\}^{K\times1}.
	\end{equation}
	Let $\mathcal{C} \in \{0,1\}^{L \times f_2(L)}$ be the set of all possible binary code sets that satisfy the UD condition (\ref{null01}), where $f_2(L)$ is the maximal possible value. As a corollary, any UD code set of $\mathcal{C}$ can be reduced to a $\mathbf{C}_{L\times K}$, where the first $L$ columns form a Hadamard code matrix and still satisfy the condition (\ref{null01}). It can be achieved by simply performing binary complement of each row or column and permuting of each rows and columns of the UD code set. 
	
	In \cite{Khachatrian1989}, the authors present a recursive construction of class of binary UD code set $\mathbf{C}_{L}$ with the size of $(L \times K_{max}^b)$. First, let us define a recursive construction of matrix $\mathbf{B}_k$ with the size of $(t \times r)$, where $r = 2^k - 1$, $t = k2^{k-1}$ and $k = 1, 2, ...$, as follows, \vspace{-0.1cm}
	\begin{eqnarray}
	\label{Bs}
	{\mathbf{B}_{k+1}}=
	\renewcommand{\baselinestretch}{1.2}
	{\footnotesize \setcounter{MaxMatrixCols}{34}
		\begin{bmatrix}
		\mathbf{B}_k&\mathbf{Z}_{t,1}&\mathbf{B}_k\\
		\mathbf{Z}_{1,r}&\mathbf{J}_{1,1}&\mathbf{J}_{1,r}\\
		\mathbf{B}_k&\mathbf{J}_{t,1}&\mathbf{B}_k'\\
		\mathbf{Z}_{r,r}&\mathbf{J}_{r,1}&\mathbf{I}_{r}
		\end{bmatrix}},
	\end{eqnarray}
	where $\mathbf{Z}_{t,m}$ matrix contains $(t \times m)$ zeros, $\mathbf{I}_r$ is the identity matrix of dimension $r$, $\mathbf{B}_k'$ is binary complement of $\mathbf{B}_k$, and $\mathbf{C}_1 = 1$. For the case when $L = 2^k-1$ then the UD code set $\mathbf{C}_{L} = \mathbf{B}_k^T$, otherwise when $2^k \leq L \leq 2^{k+1}-2$,  
	\begin{eqnarray}
	\label{Cs}
	{\mathbf{C}_{L}}=
	\renewcommand{\baselinestretch}{1.2}
	{\footnotesize \setcounter{MaxMatrixCols}{34}
		\begin{bmatrix}
		\mathbf{B}_k&\mathbf{Z}_{t,1}&\mathbf{B}_k^p\\
		\mathbf{Z}_{1,r}&\mathbf{J}_{1,1}&\mathbf{J}_{1,p}\\
		\mathbf{C}_p^r&\mathbf{J}_{T(p),1}&\mathbf{B}_k'\\
		\mathbf{Z}_{p,r}&\mathbf{J}_{p,1}&\mathbf{I}_{p}
		\end{bmatrix}}^T,
	\end{eqnarray}
	where $\mathbf{B}_k^p$ is the first $p$ columns of $\mathbf{B}_k$, and $\mathbf{C}_p^r= [\mathbf{C}_p^T \: \:  \mathbf{Z}_{\gamma(p), r-p}]_{\gamma(p) \times L+r-p}$.
	
	\vspace{0.1cm}
	\section{Minimum Distance of Code sets}
	\label{minDist}	 
	We define the minimum distance among $L$-dimensional two vectors $\mathbf{y}_i$ and $\mathbf{y}_j$ for $i \neq j$ to be
	\begin{eqnarray}
	\label{dist}
	d_L (\mathbf{y}_i, \mathbf{y}_j) = \sum_t^L | y_{i,t} -y_{j,t}|.
	\end{eqnarray}
	Then the general minimum distance of received vectors for a given code set can be formulated by
	\begin{eqnarray}
	\label{MinDistCode}
	d_{min} (\mathbf{C}) = \argmin_{\substack{\mathbf{x}_i,\mathbf{x}_j \in \{0,1\}^{K\times1} \notin \{0\}^{K\times 1} \\ \mathbf{y}_i =\mathbf{C}\mathbf{x}_i,\mathbf{y}_j =\mathbf{C}\mathbf{x}_j}} d_L (\mathbf{y}_i, \mathbf{y}_j).
	\end{eqnarray} 
	\begin{theorem}
		\label{theor1}
		Let $\mathcal{M} \in \{ 0, 1\}^{L \times K}$ represent the set of all binary matrices with distinct columns. Then $d_{min} (\mathcal{M})$ is equal to $1$.
	\end{theorem}
	\begin{proof}
		Assume that $d_{min} (\mathcal{M}) = d_L (\mathbf{y}_n, \mathbf{y}_m) $, where $\mathbf{y}_n = \mathcal{M}\mathbf{x}_{n}$ and $\mathbf{y}_m = \mathcal{M}\mathbf{x}_{m}$. The difference vector $\mathbf{y} = \mathbf{y}_n- \mathbf{y}_m = \mathcal{M}(\mathbf{x}_{n}-\mathbf{x}_{m}) = \mathcal{M}\mathbf{\bar{x}}$ must have one non-zero element $y_c \neq 0$, $y_{n,c} \neq y_{m,c}$, and $L-1$ zeros $y_t = 0$, $y_{n,t} = y_{m,t}$ for $t \neq c$  to achieve $d_{min}$. The minimum values of $y_{n,c}$ and $ y_{m,c}$ with the combination of $0, 1$ can only have $1-(0) = 1$ or $0-1 =-1$. Therefore, we can have $y_{n,c} =1$ and $y_{m,c} =0$ or $y_{n,c} =0$ and $y_{m,c} =1$, which results in both cases $d_{min} (\mathcal{M}) = |y_{n,c} - y_{m,c} | = 1$. 
	\end{proof}
	Now that we proved that $d_{min} (\mathcal{M}) = 1$, we will try to find $d_{min} (\mathbf{C})$ of UD code sets $\mathbf{C}_L \in \mathcal{C} \subset \mathcal{M}$, where $\mathcal{C} \in \{ 0,1\}^{L \times K}$ is the set of all the UD code sets. Based on constructions in (\ref{Cs}), we observe that the last column of the matrices $\mathbf{C}_L$ is $[1, 0, ..., 0]^T$. If we allow the $x_{n,K_{max}^b} \neq x_{m,K_{max}^b}$, and $x_{n,t} = x_{m,t}$ for all $t \notin \{K_{max}^b\}$ then either $y_{n,K_{max}^b} =1$ and $y_{m,K_{max}^b} =0$ or $y_{n,K_{max}^b} =0$ and $y_{m,K_{max}^b} =1$ will result in $d_L (\mathbf{y}_n, \mathbf{y}_m) = 1$. With this specific observation together with the Theorem \ref{theor1}, we conclude that all the recursive constructions in (\ref{Cs}), $d_{min} (\mathbf{C}) = 1$.
	\vspace{0.1cm}
	\section{Noiseless decoding algorithm}
	\label{NDAalgorithm}
	In the following, we describe a recursive algorithm to decode all multiplexed signals in the absence of noise. Suppose that $K$ signals contribute $\{0, 1\}$ information bits and
	\begin{equation}
	\label{system_eq}
	\mathbf{y} = \mathbf{C}\mathbf{x} = \sum_{i=1}^K \mathbf{c}_i x_i,
	\end{equation}
	\nid where $\mathbf{y} \in \mathcal{N}^{L\times 1}$, $\mathcal{N} \in \{0, 1, ..., K\}$ is the multiplexed signal vector, $\mathbf{C}\in
	\{0, 1\}^{L\times K}$ is the proposed code set,
	$\mathbf{c}_i \in \{0, 1\}^{L\times 1}$ is the $i$-th signal signature, $i = 1,\cdots,K$, and $\mathbf{x} \in \{0, 1 \}^{K\times 1}$ is the information
	bit vector. By the design of $\mathbf{C}$, (\ref{system_eq}) has the
	property that all possible $2^K$ bit-weighted sums of the
	$\mathbf{c}_i$ signatures are distinct. This means that we can
	recover $\mathbf{x}$ uniquely and correctly from $\mathbf{y}$. Let $\mathbf{y} = [\mathbf{y}_1^T, y_2, \mathbf{y}_3^T]^T$ and $\mathbf{x} = [\mathbf{x}_1^T, x_2, \mathbf{x}_3^T, \mathbf{x}_4^T]^T$, where $\mathbf{y}_1 \in \mathcal{N}^{r\times1}$, $y_2 \in \mathcal{N}$,  $\mathbf{y}_3 \in \mathcal{N}^{p\times 1}$, $\mathbf{x}_1 \in \{0,1\}^{t\times 1}$, $x_2 \in \{0,1\}$, $\mathbf{x}_3 \in \{0,1\}^{\gamma{(p)}\times 1}$, and $\mathbf{x}_4 \in \{0,1\}^{t\times 1}$. The demultiplexing NDA algorithm for the cases of $L \neq 2^k-1$ is given in direct implementation form in Table \ref{multiuser_decoding}. It is easy to modify NDA for the cases of $L= 2^k-1$.\\
	\vspace{-0.0cm}
	
	
	\begin{table}[ht]
		\vspace{-0.5cm} \caption {}
		\centering  %
		\begin{tabular}{l}
			\hline \hline \rule{0pt}{3ex} 
			\nid \textbf{Noiseless Decoding Algorithm (NDA)}  \\
			\hline \rule{0pt}{3ex} 
			\nid \textbf{{Input}:} $\mathbf{y} = 
			\mathbf{C}\mathbf{x}$ \\
			\hspace{0.3cm} 1: \hspace{0.3cm} $\mathbf{x}_1^T \mathbf{B}_k+ \mathbf{x}_3^T \mathbf{C}_p^r = \mathbf{y}_1^T \, eq.\inlineeqno $  \\
			\hspace{0.3cm} 2: \hspace{0.3cm} $x_2 \mathbf{J}_{1,1}+ \mathbf{x}_3^T \mathbf{J}_{\gamma{(p)},1} = y_2  \inlineeqno$\\ 
			\hspace{0.3cm} 3:  \hspace{0.3cm} $\mathbf{x}_1^T \mathbf{B}_k+ x_2 \mathbf{J}_{1,p} +\mathbf{x}_3^T \mathbf{C}_p' + \mathbf{x}_4^T = \mathbf{y}_3^T \inlineeqno$\\
			\hspace{0.3cm} 4:  \hspace{0.3cm} Select the first $p$ coordinates of (17). \\
			\hspace{0.3cm} 5:  \hspace{0.3cm} $\mathbf{x}_1^T \mathbf{B}_k^p+ \mathbf{x}_3^T \mathbf{C}_p = \mathbf{y}_{1,p}^T \inlineeqno$  \\
			\hspace{0.3cm} 6:  \hspace{0.3cm} Perform  $\mathbf{J}_{1,p}\otimes (18)$ operation. \\
			\hspace{0.3cm} 7:  \hspace{0.3cm} $x_2 \mathbf{J}_{1,p}+ \mathbf{x}_3^T \mathbf{J}_{\gamma{(p)},p} = \mathbf{y}_{2,p}^T \inlineeqno$ \\
			\hspace{0.3cm} 8:  \hspace{0.3cm} Add $(19)$ to $(20)$ and subtract $(21)$. \\
			\hspace{0.3cm} 9:  \hspace{0.3cm}  2$\mathbf{x}_1^T \mathbf{B}_k^p+ \mathbf{x}_4^T = \mathbf{y}_{1,p}^T + \mathbf{y}_3^T  -  \mathbf{y}_{2,p}^T \inlineeqno$  \\
			\hspace{0.3cm}10:  \hspace{0.3cm} Decode $\mathbf{x}_4$ uniquely from $(22)$.  \\
			\hspace{0.3cm}11:  \hspace{0.3cm} $\mathbf{x}_1^T \mathbf{B}_k^p = \frac{1}{2}(\mathbf{y}_{1,p}^T +\mathbf{y}_3^T  -  \mathbf{y}_{2,p}^T -\mathbf{x}_4^T)$ \\
			\hspace{0.3cm}12:  \hspace{0.3cm} Using induction solve for $\mathbf{x}_1$; \\
			\hspace{0.3cm}13:  \hspace{0.3cm} Substitute $\mathbf{x}_1$ in $(20)$ to get $\mathbf{x}_3$. \\
			\hspace{0.3cm}14: \hspace{0.3cm} Substitute $\mathbf{x}_3$ in $(18)$ to get $x_2$. \\
			\nid \textbf{{Output}:} $\mathbf{{x}}$ \\
			\hline
		\end{tabular}\vspace{-0.0cm}
		\label{multiuser_decoding}
	\end{table}
	
	\section{Fast Decoding Algorithm in AWGN}
	\label{fastDecoder}
	The recursive linear NDA decoder discussed in Section \ref{NDAalgorithm} is not suitable for the noisy transmission channel. The received vector form in the presence of noise can be expressed as
	\begin{eqnarray}
	\mathbf{y} &=& A\mathbf{C}\mathbf{x} + \mathbf{n} \\
	&=& \sum_{j=1}^K A\mathbf{c}_j x_j + \mathbf{n},
	\end{eqnarray}
	\noindent where $A$ is the amplitude, $\mathbf{c}_j \in \{0, 1\}^{L\times 1}$ are signatures for $1\leq j \leq K$, $\mathbf{x} \in \{0, 1\}^{K\times 1}$ is user data and $\mathbf{n} \in \mathbb{R}^{L \times 1}$ is the additive white Gaussian noise (AWGN) channel noise vector. The objective of the receiver is the following; given the received vector $\mathbf{y}$ and $\mathbf{C}$ recover the user data $\mathbf{\hat{x}}$ such that the mean square error $E\{||\mathbf{x}-\mathbf{\hat{x}}||^2\} $ is minimized. It is known that obtaining the ML solution is generally NP-hard \cite{Lupas1989}.

	For our detection problem, where the overloaded signature matrix has UD structure, can be solved efficiently if there is a function that maps $\mathbf{y} \mapsto \widehat{\mathbf{y}} \in \Lambda \subset \mathcal{N}^{L\times 1}$, where $\Lambda$ is a $\mathbb{Z}$-module\footnote {A module over a ring $R$ is an Abelian group $M$, which can be considered as a generalization of the notion of vector space over a field.} with rank $L$. It is equivalent to finding the closest point in a lattice $\Lambda$, such that 
	\begin{eqnarray}
	\label{MinDistY}
	\widehat{\mathbf{y}} = \argmin_{\mathbf{y}' \in \mathcal{N}^{L\times 1}} d_L (\mathbf{y}, \mathbf{y}').
	\end{eqnarray}
	Gaining the knowledge of $\widehat{\mathbf{y}}$, one of the points in $\Lambda$ generated by $\mathbf{C}$, we can obtain  $\mathbf{\hat{x}}$ uniquely using NDA algorithm, since $\mathbf{C}$ satisfies the uniquely decodability criteria (\ref{null01}). However, there is no known polynomial algorithm to obtain $\widehat{\mathbf{y}}$ from $\mathbf{y}$. 
	
	Without loss of generality, our proposed simplified ML approach uses the fact that the first row of the $\mathbf{C}_L$ under consideration is all ones. Though it does not necessarily imply that our proposed fast decoder cannot be applied to the recursive UD code sets constructed by (\ref{Cs}). It will only require a slight modification  such as binary complementing, permuting rows and columns. We are ready to present the general form of the fast decoding algorithm (FDA) for the $\mathbf{C}_L$, $L = 2^k -1$ and $L \neq 2^k -1$ cases, where the quantizer $Q : \mathbb{R} \mapsto \mathcal{N} $,  $z_1 = Q(y, 0, K)$ is a mapping of $y \in \mathbb{R}$ to the constellation of $\{0, 1, ..., K\}$. 
	\vspace{-0.0cm}
	\begin{table}[ht]
		\caption {} 
		\vspace{-0.5cm}
		\begin{center}
			\begin{tabular}{l}
				\hline \hline \rule{0pt}{3ex} 
				\nid \textbf{Fast Decoder Algorithm (FDA)}  \\
				\hline \rule{0pt}{3ex} 
				\nid \textbf{{Input}:} $\mathbf{y}$ \\
				\hspace{0.3cm} 1: $z_1 \gets Q(y_1, 0,K)$ \\
				\hspace{0.3cm} 2: \textbf{if} $|z_1| = K$,   $\mathbf{\hat{x}}\gets \mathbf{1}_K$\\ 
				\hspace{0.3cm} 3:  \textbf{else}\\
				\hspace{0.3cm} 4:  \hspace{0.3cm} $\mathbf{m} \gets -\mathbf{1}_K$, $r_c \gets 1$, $n \gets z_1$\\
				\hspace{0.3cm} 5:  \hspace{0.3cm} $m_{LR}(r_c,3) \gets n$\\
				\hspace{0.3cm} 6:  \hspace{0.3cm} $dP(r_c) \gets [n, K, m_{LR}(r_c,1), m_{LR}(r_c, 2), mP(r_c)]$\\
				\hspace{0.3cm} 7:  \hspace{0.3cm} $\mathbf{c}_{AL} \gets \mathbf{0}$, $\mathbf{z} \gets \mathbf{0}$, $s_{I}\gets 1$, $c_{T} \gets 1$\\
				\hspace{0.3cm} 8:  \hspace{0.3cm} \textbf{while} ($s_{I}=1$ AND $c_{T}< N_c$)\\
				\hspace{0.3cm} 9:  \hspace{0.6cm} $s_{I}\gets 0$\\
				\hspace{0.3cm}10:  \hspace{0.6cm} \textbf{while} ($r_c < L$, $r_c \gets r_c +1$)\\
				\hspace{0.3cm}11:  \hspace{0.8cm} $[dP(r_c),\!m]\!\! \gets\!\! meP(dP(r_c-1), \mathbf{m}, n, K, r_c, m_{LR}, mP)$\\
				\hspace{0.3cm}12:  \hspace{0.8cm} $A_{min} \gets minT(dP(r_c))$, $A_{max} \gets maxT(dP(r_c))$\\
				\hspace{0.3cm}13:  \hspace{0.8cm} $\mathbf{z}(r_c) \gets Q(y', A_{min}, A_{max},1)$\\
				\hspace{0.3cm}14:  \hspace{0.8cm} $\mathbf{c}_{AL}(r_c,2) \gets (A_{min}- A_{max}) +1$\\
				\hspace{0.3cm}15:  \hspace{0.8cm} $m_{LR}(r_c,3) \gets \mathbf{z}(r_c)$\\
				\hspace{0.3cm}16:  \hspace{0.8cm} $m_{LR}(r_c,4) \gets n- m_{LR}(r_c,3)$ \\
				\hspace{0.3cm}17:  \hspace{0.8cm} $\mathbf{m} \gets uM(\mathbf{m}, m_{LR}, r_c, mP)$ \\
				\hspace{0.3cm}18:  \hspace{0.6cm} $\mathbf{m} \gets f_c(\mathbf{m}, m_{LR})$, $\mathbf{t}_D \gets \mathbf{z} - \mathbf{C}\mathbf{m}$\\
				\hspace{0.3cm}19:  \hspace{0.6cm}  \textbf{if} $\mathbf{t}_D \notin \mathbf{0}$,   $s_{I}\gets 1$, $r_c \gets i_d$   \\
				\hspace{0.3cm}20:  \hspace{0.9cm} $\mathbf{c}_{AL}(r_c+1,1) \gets  \mathbf{c}_{AL}(r_c+1,1)+1$ \\ 
				\hspace{0.3cm}21:  \hspace{0.9cm} $c_{T} \gets c_T +1$ \\
				\hspace{0.3cm}22:  \hspace{0.3cm} $\mathbf{\hat{x}} \gets \mathbf{m} $  \\
				\nid \textbf{{Output}:} $\mathbf{\hat{x}}$ \\
				\hline
			\end{tabular}\vspace{-0.0cm}
		\end{center}
	\end{table}
	\newline
	Furthermore, let $n$ and vector $\mathbf{m}$ denote the number of $+1$s and locations in $\mathbf{\hat{x}}$, respectively. Note that when $z_1 = K$ only one comparison is required. The algorithm proceeds by partitioning each row $dP(r_c)$, recording $n$, $K'$ the length of partition, $L'$ and $R'$ are lengths of $+1$s and $0$s in that specific partition of the row. Whereas the table $m_{LR}$ keeps track of $L$ and $R$, which are the lengths of $+1$s and $0$s of the row, $nL$, $nR$, the number of $+1$s in the $+1$s and $0$s locations of the row, respectively, and $mP(r_c)$ is the actual column indices of $+1$s and $0$s at each row. The function $meP(dP(r_c-1), \mathbf{m}, n, K, r_c, m_{LR}, mP)$ scans each partition of the row with updated values and if it finds one or more partitions that hits the boundaries it will partition further or completely define, in other words, it knows the exact locations of $0$s in that partition. $A_{min}$ and $A_{max}$ are minimum and maximum calculated given partitions at each row, in line $13$ of FDA algorithm, we define $y'= y(r_c)+2\mathsf{sgn}(y(r_c)-z(r_c))c_{AL}(r_c,1)$, and in line $17$ $uM(\mathbf{m}, m_{LR}, r_c, mP)$ updates $\mathbf{m}$ with the given updated parameters and function $f_c(\mathbf{m}, m_{LR})$ re-calculates all the locations of $+1$s and $0$s in $\mathbf{m}$ based on the UD structure of $\mathbf{C}$ codes. In case the information in $\mathbf{m}$ do not correspond to $\mathbf{z}$, which is verified in line $19$ then it sets the $r_c$ to index $i_d$ in which $\mathbf{t}_D(i_d) \neq 0$, where the discrepancy happened and re-runs from line $10$ until it finds $\mathbf{m}$ that correspondence to $\mathbf{z}$.  
	
	\section{Complexity Analysis}
	\label{performanceAnalysis}
	In this section, we discuss the complexity analysis of the proposed NDA and FDA algorithms. The NDA decoder for the noiseless transmission channels discussed in Section \ref{NDAalgorithm}, deciphers all the users data at the receiver side in a recursive manner. At each step it performs addition and comparisons to decipher the bits for the users. 
	From the equation ($22$), we can uniquely identify $\mathbf{x}_4$, since it is the only vector that contributes to the resultant vector to be odd. Given the right hand side of ($22$) if any element is odd that means the same element in $\mathbf{x}_4$ is $+1$ otherwise it is $0$. Then using induction method it is possible to obtain $\mathbf{x}_1$. Finally, substituting $\mathbf{x}_1$ into $(20)$ to obtain $\mathbf{x}_3$, after which substituting $\mathbf{x}_3$ into $(18)$ to obtain $x_2$, respectively. The algorithm returns all the decoded $K$ bits, which results in linear complexity $\mathcal{O}(K)$. Generally, one would accept that the complexity of a decoder in noisy channels is much higher than in noiseless channels. However, the complexity of the proposed FDA decoder presented in Section \ref{fastDecoder} is not any worse than NDA in terms of the Big-$\mathcal{O}$ notation. It is important to state that the beauty of the proposed FDA lies in the fact that it neither requires any matrix inversion nor decomposition, instead only a few comparisons are performed in the quantizer $Q(\cdot)$, i.e., multiplications and additions. The algorithm goes through each row of the received vector to decode one or more users. Unlike the noiseless transmission FDA algorithm can repeat the decoding process again from a row that was previously been decoded by $N_c$ times to improve the results. As a result, it performs $LN_c$ times instead of $L$. Therefore, the average complexity of FDA algorithm still remains linear in $K$, i.e., $\mathcal{O}(K)$, since $N_c$ is simply a constant.

	\section{Simulation results}
	\label{simulation}
	In this section, we evaluate the performance of the proposed antipodal UD code sequences generated by (\ref{Cs}), which are shown in Figs. \ref{C_4_5_01} and \ref{C_8_13_02}. 
	\begin{figure}[h]
		\centering
		\begin{center}
			\begin{eqnarray}
			{\bf C}_{4\times 5}=
			\renewcommand{\baselinestretch}{1}
			{\tiny \setcounter{MaxMatrixCols}{34}
				\begin{bmatrix}
				0& \m 0&\m 0&\m 0&\m 1\\
				1&\m 1&\m 0&\m 1&\m 0\\
				0&\m 1&\m 1&\m 0&\m 0\\
				1&\m 0&\m 1&\m 0&\m 0\\
				\end{bmatrix}}
			\nonumber
			\end{eqnarray}
			\caption{UD code set $\mathbf{C}$ with $L=4$ and $K=5$.} \label{C_4_5_01}
		\end{center}
	\end{figure}
	\begin{figure}[h]
		\begin{eqnarray}
		{\bf C}_{8\times 13}=
		\renewcommand{\baselinestretch}{1}
		{\tiny \setcounter{MaxMatrixCols}{34}
			\begin{bmatrix}
			0&\m 0&\m 0&\m 0&\m 0&\m 0&\m 0&\m 0&\m 0&\m 0&\m 0&\m 0&\m 1\\
			1&\m 1&\m 0&\m 1&\m 1&\m 0&\m 0&\m 1&\m 0&\m 0&\m 0&\m 1&\m 0\\
			0&\m 1&\m 1&\m 0&\m 1&\m 1&\m 0&\m 0&\m 1&\m 0&\m 1&\m 0&\m 0\\
			1&\m 0&\m 1&\m 0&\m 1&\m 0&\m 1&\m 0&\m 1&\m 1&\m 0&\m 0&\m 0\\
			0&\m 0&\m 0&\m 0&\m 1&\m 1&\m 1&\m 1&\m 1&\m 0&\m 0&\m 0&\m 0\\
			1&\m 1&\m 0&\m 1&\m 0&\m 1&\m 1&\m 0&\m 1&\m 0&\m 0&\m 0&\m 0\\
			0&\m 1&\m 1&\m 0&\m 0&\m 0&\m 1&\m 1&\m 0&\m 0&\m 0&\m 0&\m 0\\
			1&\m 0&\m 1&\m 0&\m 0&\m 1&\m 0&\m 1&\m 0&\m 0&\m 0&\m 0&\m 0\\
			\end{bmatrix}}
		\nonumber
		\end{eqnarray}
		\caption{UD code set $\mathbf{C}$ with $L=8$ and $K=13$.} \label{C_8_13_02}
	\end{figure}

	In our simulations, we compare the proposed decoder with ML decoder and the PDA proposed in \cite{Romano2005} to decipher the proposed code set sequences. The comparison in our simulation is performed with PDA algorithm alone as it has the best performance compared to other decoding algorithms (e.g., MF, MMSE, PIC, etc.).  In Figs. \ref{C4by5} and \ref{C8by13}, we plot the BER performance averaged over the different users for $\mathbf{C}_{4 \times 5}$ and $\mathbf{C}_{8\times 13}$, respectively. As we can see from Figs. \ref{C4by5} and \ref{C8by13} for a BER of $10^{-3}$ the performance of the proposed detector is only about 0.2 dB and 1 dB inferior compared to the ML decoder for UD code sets $\mathbf{C}_{4 \times 5}$ and $\mathbf{C}_{8 \times 13}$, respectively. While the performance of the PDA suffers significantly compared to our purposed decoder. Although the BER performance of the proposed decoder is slightly higher than that of the ML detector, it is much less complex and less costly to implement compared to the ML decoder.

	\vspace{-0.2cm}
	\begin{center}
		\begin{figure}[h]
			\centering
			\includegraphics[width=3.85 in]{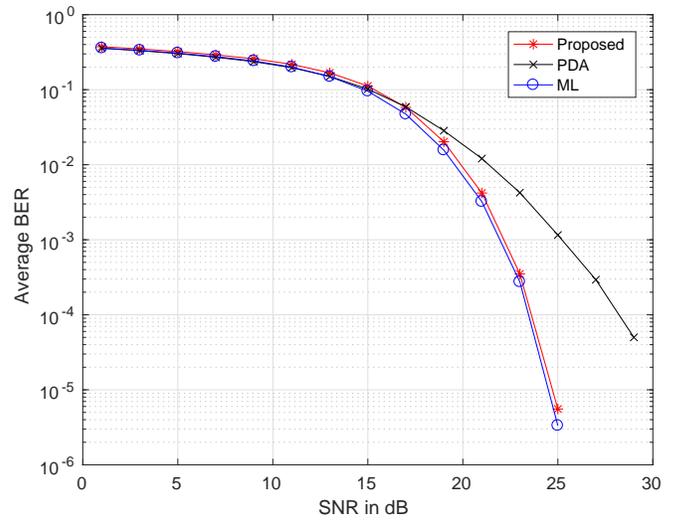}
			\centering \caption{UD code set $\mathbf{C}_{4 \times 5}$.} \label{C4by5}
		\end{figure}
	\end{center}
	\vspace{-0.4cm}
	\begin{center}
		\begin{figure}[h]\vspace{-0.9cm}	
			\centering
			\includegraphics[width=3.85 in]{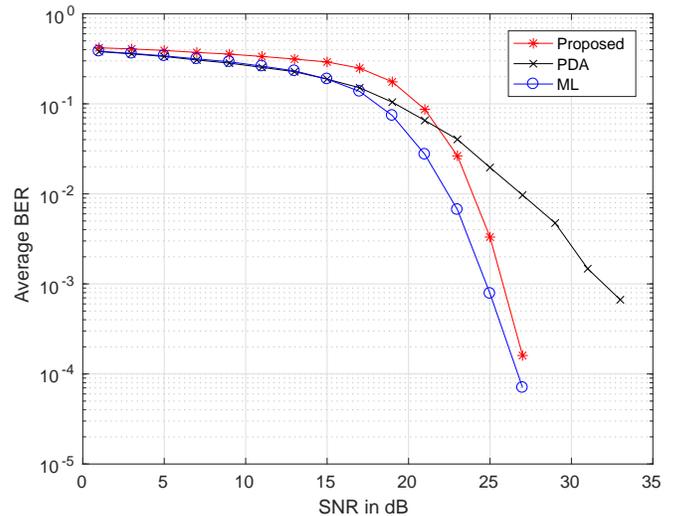}
			\centering \caption{UD code set $\mathbf{C}_{8 \times 13}$.} \label{C8by13}
		\end{figure}
	\end{center}\vspace{-0.4cm}
	\section{Conclusion}
	\label{conclusion}
	In this paper, we have introduced a novel fast decoder algorithm (FDA) for  uniquely decodable (UD) code sets for overloaded synchronous
	optical code-division multiple-access (O-CDMA). The proposed simple decoder uses only a few comparisons and can allow the user to uniquely recover the information bits at the receiver side. The proposed algorithm has much lower computational complexity compared to maximum-likelihood (ML) decoder, which has a high complexity even for moderate code lengths. Simulation results show that the performance of the proposed decoder is almost as good as the ML decoder. 
	

\end{document}